\documentclass[11pt, a4paper, oneside]{article}

\usepackage{etoolbox}

\usepackage[sorting=none, giveninits=true, doi=true, url=false, maxnames=6, pagetracker=false
]{biblatex}
\AtEveryBibitem{%
  \clearfield{issn}%
  \clearfield{month}%
  \clearfield{pages}%
  \clearlist{month}%
}
\bibliography{ref}

\usepackage{amsfonts}
\usepackage{booktabs}
\usepackage{tabularx}
\usepackage{xcolor}
\usepackage{arydshln}
\usepackage{array}
\usepackage{caption}
\usepackage{enumitem}
\usepackage{multicol}
\usepackage{amssymb}
\usepackage{setspace}
\usepackage{graphicx}
\usepackage{graphics}
\usepackage{wrapfig}
\usepackage{amsmath}
\usepackage{tabularx}
\usepackage{fontspec}
\usepackage[dvipsnames]{xcolor}
\usepackage{amsthm}
\usepackage{newtxtext, newtxmath}
\usepackage{tikz}
\usepackage{sidecap}
\usepackage{algorithm}
\usepackage{algpseudocode}
\usepackage{listings}
\usepackage{multirow}
\usepackage{xcolor}
\usepackage{xurl}
\definecolor{RKeyword}{RGB}{0, 0, 180}
\definecolor{RComment}{RGB}{0, 140, 0}
\definecolor{RString}{RGB}{180, 100, 0}
\lstdefinestyle{rlanguage}{
  language          = R,
  basicstyle        = \ttfamily\small,
  keywordstyle      = \color{RKeyword}\bfseries,
  commentstyle      = \color{RComment},
  stringstyle       = \color{RString},
  showstringspaces  = false,
  breaklines        = true,
  columns           = fullflexible,
  frame             = tb,
  framerule         = 0.4pt,
  rulecolor         = \color{black},
  captionpos        = b,
  aboveskip         = 1em,
  belowskip         = 1em,
  numbers           = none
}

\usetikzlibrary{arrows.meta}
\usetikzlibrary{positioning,calc,fit,arrows.meta,shapes.multipart,calc,backgrounds,fit}
\usepackage{tikzlings}
\usepackage{diagbox}
\usepackage[table]{xcolor}

\usepackage{geometry}
 \usepackage[hidelinks]{hyperref}
 
\usepackage{array}
\newcolumntype{L}[1]{>{\raggedright\arraybackslash}p{#1}}
\newcolumntype{C}[1]{>{\centering\arraybackslash}p{#1}}
\newcolumntype{R}[1]{>{\raggedleft\arraybackslash}p{#1}}

\newtheorem{definition}{Definition}

\newtheorem{proposition}{Proposition}

\newtheorem{assumption}{Assumption}
\newtheorem{corollary}{Corollary} 
\usepackage{fontspec}

\usepackage[utf8]{inputenc}
\usepackage{mathtools}
\usepackage{tcolorbox}
\usepackage{float}
\usepackage{tabularx}
\usepackage[normalem]{ulem}
\definecolor{arsenic}{rgb}{0.23, 0.27, 0.29}
\definecolor{burntorange}{rgb}{0.8784,0.8863,0.9059}
\definecolor{burntorange2}{rgb}{0.1725,0.2118,0.4157}
\usepackage[font=small,skip=0pt]{caption}

\definecolor{arsenic}{rgb}{0.23, 0.27, 0.29}
\definecolor{burntorange}{rgb}{0.8, 0.33, 0.0}
\definecolor{royalblue}{rgb}{0.15, 0.30, 0.60}

\definecolor{softgreen}{RGB}{110,170,120}
\definecolor{softorange}{RGB}{230,150,80}
\definecolor{softred}{RGB}{210,90,90}
\definecolor{softpurple}{RGB}{140,120,200}

\usepackage{colortbl}
\usepackage{xcolor}

\usepackage{titlesec}
\titleformat{\section}{\normalfont\large\bfseries}{\thesection}{1em}{}

\titleformat{\subsection}
  {\normalfont\normalsize\bfseries}{\thesubsection}{1em}{}

\usepackage{fancyhdr}
\usepackage{etoolbox}

\usepackage{pgfplots}
\pgfplotsset{compat=1.18}
\usetikzlibrary{positioning, calc, arrows.meta, fit}

\pgfplotsset{
  heatmap style/.style={
    width=4cm,
    height=4cm,
    view={0}{90},
    colormap={mycmap}{color=(white) color=(black)},
    colorbar=false,
    axis on top,
    axis line style={draw=black},
    xtick={-1, 0, 1},
    ytick={-1, 0, 1},
    tick label style={font=\scriptsize},
    tick style={color=black},
    scale only axis,
    enlargelimits=false,
    xmin=-1, xmax=1, ymin=-1, ymax=1,
    point meta min=0, point meta max=1,
    title style={font=\small, yshift=-2pt},
    xlabel={$x_1$},
    xlabel style={font=\small, yshift=2pt},
    ylabel style={font=\small, yshift=-2pt, xshift=0pt}
  },
  regression style/.style={
    width=4cm,
    height=4cm,
    axis lines=box,
    axis line style={draw=black},
    xtick={-2, 0, 2},
    ytick={-2, 0, 2},
    tick label style={font=\scriptsize},
    tick style={color=black},
    scale only axis,
    xmin=-2.5, xmax=2.5, ymin=-2.5, ymax=2.5,
    title style={font=\small, yshift=-2pt},
    xlabel={$x$},
    xlabel style={font=\small, yshift=2pt},
    ylabel style={font=\small, yshift=-2pt, xshift=6pt}
  }
}

\begin{document}

{\Large
\noindent The MCMC convergence graph: a diagnostic for multimodal posteriors
}

\begin{flushleft}
\bigskip 
Cici Chen Gu\textsuperscript{1 $\circledast$}, 
Eszter Lakatos\textsuperscript{2, 3}, 
Sara Hamis\textsuperscript{1 $\circledast$}.

\bigskip
\textbf{1} Department of Information Technology, Uppsala University, Uppsala, Sweden.\\ 
\textbf{2} Department of Mathematical Sciences, Chalmers University of Technology, Gothenburg, Sweden.
\\
\textbf{3} Department of Mathematical Sciences, University of Gothenburg, Gothenburg, Sweden.\\

\end{flushleft}

\smallskip
\noindent{\bf Author ORCiDs}: 
\smallskip
\textbf{CCG:} 0009-0008-1998-6486; 
\textbf{EL:} 0000-0002-7221-6850; 
\textbf{SH:} 0000-0002-1105-8078.

\bigskip
\noindent{\bf $\circledast$ Corresponding author}: chen.gu@it.uu.se, sara.hamis@it.uu.se.


\thispagestyle{empty} 
\vspace{5cm}

\section*{Abstract}
Multimodality is common in many scientific and engineering applications. 
However, standard Markov chain Monte Carlo (MCMC) convergence diagnostics often penalise multimodality in practice by conflating overall poor mixing with groups of chains that have mixed well within distinct modes. 
To diagnose multimodality rather than flag it as a sampling failure, we introduce the MCMC convergence graph: 
a graph-based diagnostic that computes pairwise $\hat{R}$ values between chains and summarises them in a graph, revealing sets of chains that explore the same posterior mode. 
We implement the method in the probabilistic programming framework Stan and demonstrate its use on worked examples ranging from regression models with inherent multimodality in the data, to a pharmacokinetic model for which multimodality reflects non-identifiability. 
The graph-based diagnostic, which we make freely available in the R package \texttt{mcmcConvergenceGraph} and which accepts MCMC output
from any sampler, provides both graphical and numerical summaries of MCMC behaviour for diagnosing multimodality.

\newpage

\section{Introduction}
\label{sec:introduction}
Multimodality is prevalent in many scientific and engineering applications, including biology \cite{vergnon2012}, psychology \cite{haslbeck2023}, economics \cite{adrian2021}, and sociology \cite{downey2001}. 
Multimodality typically arises from the data itself or from non-identifiability inherent to the model.
When multimodal posteriors are targeted with multiple Markov chain Monte Carlo (MCMC) chains~\cite{gilks1995}, groups of chains that separate across distinct modes can produce chain disagreement indistinguishable from that of chains that simply fail to mix~\cite{yao2022}. 
Accordingly, multimodality is routinely flagged as a sampling pathology.

Existing approaches to addressing multimodality in MCMC fall into two broad directions: 
improving samplers to explore complex posterior geometries more efficiently, 
and improving convergence diagnostics that assess the resulting MCMC output. 
The first direction recognises that, 
even though asymptotic convergence to the target is guaranteed under regularity conditions~\cite{roberts1994},
in practice MCMC produces only a finite sample whose empirical distribution \emph{approximates} the target. 
For a multimodal target separated by low-density regions, transitions between modes are unlikely to be accepted, so a chain that starts near one mode tends to remain there~\cite{yao2025}, and multiple chains can
appear individually well-mixed while collectively missing or misweighting modes~\cite{yao2022}. 
To address this, MCMC methods with better transition dynamics have been developed, 
from sampler-level improvements such as Hamiltonian Monte Carlo~\cite{duane1987,neal2011} to multi-chain strategies such as parallel tempering~\cite{geyer1991}. 
However, these methods can miss modes when the posterior geometry is sufficiently challenging: 
even sophisticated samplers may leave well-separated modes unexplored~\cite{latuszynski2025,gallegos2026}.

The second direction develops convergence diagnostics for the MCMC output~\cite{roy2020}. 
The most widely used diagnostic is Gelman and Rubin's $\hat{R}$~\cite{gelman1992}, which compares within-chain and between-chain variability.
Its extensions include 
split-$\hat{R}$~\cite{gelman2013}, 
rank-normalised split-$\hat{R}$ and rank-normalised folded split-$\hat{R}$~\cite{vehtari2021}, 
and $\hat{R}_\infty$, which assesses convergence across the quantiles of the target
distribution~\cite{moins2025}.
Other diagnostics take different approaches. 
The effective sample size~\cite{vehtari2021} estimates the number of effectively independent draws across all chains. 
The classifier-based $R^*$~\cite{lambert2022} trains a classifier to predict which chain a draw came from, signalling convergence when the chains are indistinguishable. 
Multi-objective approaches summarise multiple sampling metrics to assess convergence~\cite{kavianihamedani2024}. 
Cluster-based methods apply clustering algorithms such as k-means to group chains by their sampled values, distinguishing modes in the target~\cite{zhu2021clustering_preprint}. 

Variants of $\hat{R}$ and the effective sample size are standard outputs of widely used probabilistic programming frameworks and samplers (e.g., Stan~\cite{stan2026}, PyMC~\cite{pymc2023}, OpenBUGS~\cite{sturtz2005}) and diagnostic libraries (e.g., ArviZ~\cite{martin2026}, coda~\cite{plummer2006}).
These standard outputs flag failures of convergence, but they are not designed to characterise the structure of that failure. 
A large $\hat{R}$, for example, indicates chain disagreement without describing how the disagreement is organised:
(i) chains may have converged locally to different modes of a multimodal posterior, 
(ii) most chains may explore the same region while a few remain isolated, 
or (iii) all chains may mix poorly on a challenging unimodal target. 
No single summary can distinguish these cases, or identify multimodality specifically, and an absence of evidence for non-convergence is not evidence of convergence~\cite{vats2021}.
To address this, numerical and visual diagnostics together can help reveal the structure of chain disagreement, as advocated by Vehtari et al.\ and their discussants~\cite{vehtari2021rejoinder,hans2021discussion}.

In this spirit, we introduce the MCMC convergence graph: a graph-based diagnostic that builds on pairwise $\hat{R}$ values and summarises chain agreement both numerically and visually. 
It distinguishes three interpretable patterns: 
all chains mixing together, 
chains partitioning into distinct posterior regions, 
and one or more chains being isolated from all others. 
We make the diagnostic accessible through the freely available R package \texttt{mcmcConvergenceGraph}. 
The remainder of this paper is organised as follows. 
Section~\ref{sec:theory} presents the theory of the diagnostic. 
Section~\ref{sec:implementation} describes the R package and the implementation of a selection of worked examples, whose results are shown in Section~\ref{sec:results}. 
Section~\ref{sec:discussion} discusses practical considerations and possible extensions, and Section~\ref{sec:conclusion} concludes our work.

\section{Theory}
\label{sec:theory}
We develop a graph-based diagnostic for MCMC convergence and multimodality from pairwise $\hat{R}$ values between chains.
The concept is illustrated in Fig.~\ref{fig:pipeline}. 
We first present the theory for one-dimensional posteriors, or targets more generally, in Sections~\ref{sec:pairwise-rhat}--\ref{sec:chaincountmode}, and then extend it to arbitrary dimension in Section~\ref{sec:multidim}.

\subsection{Pairwise $\hat{R}$}
\label{sec:pairwise-rhat}
Let $\pi$ be a target distribution and let $C = \{1, 2, \ldots, n\}$ index $n$ MCMC chains targeting $\pi$. 
Each chain $i \in C$ produces $N$ samples $\{\theta_i^{(t)}\}_{t=1}^N$ drawn from a distribution $\pi_i = \pi|_{A_i}$, where $A_i \subseteq \mathrm{supp}(\pi)$ is the explored region of chain $i$. 
Throughout, we use the term \emph{mode} to denote a well-separated region of $\pi$ that a chain does not leave within a feasible run, which may itself contain more than one local maximum.
If a chain can move between two local maxima within a feasible run, they belong to the same mode. 
Thus, for a chain $i$ trapped in a single mode of $\pi$, $\pi_i$ is $\pi$ restricted to that mode. 
In theory, an ergodic chain $i$ explores all of $\pi$ as $N \to \infty$, so that $A_i = \mathrm{supp}(\pi)$ and $\pi_i = \pi$.
However, in practice, the expected time to move between modes far exceeds any feasible run length.
We therefore treat each chain $i$ as confined to a single mode, effectively sampling its restricted target $\pi_i = \pi|_{A_i}$, and study the limit $N \to \infty$ in which it converges to $\pi_i$ rather than to $\pi$.
\\

\noindent For each pair of chains $\{i,j\}$, let $\hat{R}_{ij}$ be an $\hat{R}$ statistic computed from chains $i$ and $j$ alone. 
As $N \to \infty$,
\begin{equation}
  \hat{R}_{ij} \in [1, \infty),
\end{equation}
with values close to 1 indicating agreement between the two chains.
This statistic may be the pairwise version of the split-$\hat{R}$~\cite{gelman2013}, 
the rank-normalised $\hat{R}$ reported by Stan-$\hat{R}$~\cite{vehtari2021}, 
or any other $\hat{R}$ variant satisfying properties (P1)--(P3):
\begin{enumerate}[label=(P\arabic*),leftmargin=*,labelindent=\parindent,itemsep=2pt]
  \item \emph{Symmetry.} $\hat{R}_{ij} = \hat{R}_{ji}$.
  \item \emph{Agreement.} $\hat{R}_{ij} \xrightarrow{p} 1$ as $N \to \infty$ if $\pi_i = \pi_j$.
  \item \emph{Disagreement.} $\hat{R}_{ij} \xrightarrow{p} \lambda_{ij} > 1$ as $N \to \infty$ if $\pi_i \neq \pi_j$.
\end{enumerate}
Here, $\lambda_{ij}$ is a pair-specific limiting value. 
Properties~(P1)--(P3) follow the standard consistency behaviour of $\hat{R}$~\cite{gelman1992, vehtari2021}. 
By the symmetry of (P1), we consider only pairs with $i < j$ and exclude self-pairs, giving $\binom{n}{2}$ pairwise statistics. 

\vspace{.1cm}

\begin{figure}[H]
\centering
\includegraphics[
  width=\textwidth
]{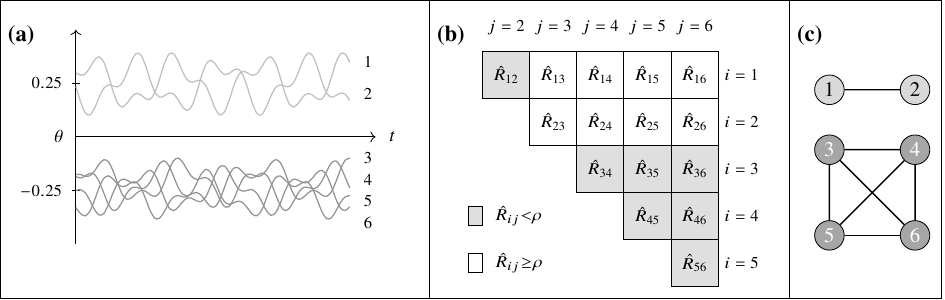}
\caption{
\textbf{Pairwise $\hat{R}$ values induce a graph-based diagnostic for MCMC convergence and multimodality.}
(a) Six MCMC chains after warmup over sample iteration $t$.
Chains 1, 2 explore one posterior region and chains 3--6 explore another. 
(b) Pairwise statistics $\hat{R}_{ij}$ for $i < j$.
Shaded cells indicate $\hat{R}_{ij} < \rho$, 
where $\rho > 1$ is the convergence threshold.
(c) The MCMC convergence graph $G_\rho$ takes the chains as nodes
and adds an edge $\{i,j\}$ whenever $\hat{R}_{ij} < \rho$. 
Its connected components group the chains by the region of the posterior they explore.
}
\label{fig:pipeline}
\end{figure}

\subsection{Constructing the MCMC convergence graph}
\label{sec:graph-const}
We define the MCMC convergence graph $G_\rho = (C, E_\rho)$ as the undirected graph in which MCMC chains are nodes and two nodes $i,j$ are connected by an edge if and only if $\hat{R}_{ij} < \rho$. 
The edge set is
\begin{equation}
  E_\rho = \{\{i,j\} \subseteq C : i < j, \; \hat{R}_{ij} < \rho\}.
\end{equation}
Two chains belong to the same connected component of $G_\rho$ if and only if they are joined by a path in $G_\rho$. 
Let $\mathcal{Q} = \{Q_1, Q_2, \ldots, Q_K\}$ denote the connected components that contain at least two chains, with $K = |\mathcal{Q}|$. 
Separately, let $\mathcal{I} \subseteq C$ denote the isolated chains, each disagreeing with every other chain and forming a singleton component, with $I = |\mathcal{I}|$. 
Hence $K$ counts the groups of mutually agreeing chains and $I$ the chains that agree with none, so $K = 0$ when every chain is isolated. 
By (P2)--(P3), the value $1$ separates agreement between chains from disagreement in the limit $N \to \infty$.
In practical applications, $N$ is finite and $\hat{R}_{ij}$ can be above $1$ 
even for well-mixed chain pairs. 
We therefore fix $\rho$ slightly above 1, adopting standard $\hat{R}$ thresholds: 
the classical $\rho = 1.1$~\cite{gelman1992} and the stricter $\rho = 1.05$ or $\rho = 1.01$~\cite{vehtari2021}.

\subsection{Reading the MCMC convergence graph}
\label{sec:graph-read}
The pair $(K, I)$ summarises the qualitative structure of $G_\rho$ and admits three regimes (R1)--(R3):
\begin{enumerate}[label=(R\arabic*),leftmargin=*,labelindent=\parindent,itemsep=4pt]
\item \emph{Global agreement} ($K = 1$, $I = 0$): 
all chains form a single connected component, 
consistent with a single explored mode or with chains that have fully explored a multimodal posterior.
\item \emph{Local agreement} ($K \geq 2$, $I = 0$): 
the chains partition into connected components of size $\geq 2$, 
connected within each component and with no agreement across components. 
Each component corresponds to a distinct mode explored by multiple chains.
\item \emph{Isolated chains} ($I \geq 1$): 
at least one chain is in disagreement with every other chain so that the graph contains isolated nodes, 
consistent with either a mode explored by a single chain or a chain that has failed to mix.
\end{enumerate}
\noindent The ratio $I/n$ quantifies the proportion of isolated chains.
A value of $I/n = 0$ indicates that every chain is in agreement with at least one other, 
whereas values close to 1 indicate widespread isolation, consistent with poor mixing.

\subsection{Asymptotic mode identification}
\label{sec:asymptotic}
We show two results.
First, as the per-chain sample size $N \to \infty$, the connected components of $G_\rho$ partition chains by their explored mode. 
Second, when no chain is isolated, $K$ identifies the number of distinct explored modes.
Both rely on the following assumption on the mode structure of $\pi$. 
\begin{assumption}[Modes and explored regions]
\label{assumption:modes}
The target $\pi$ has $M$ disjoint modes, and each chain's explored region $A_i$ is one of these modes.
\end{assumption}

\begin{proposition}[Asymptotic mode identification]
\label{prop:asymptotic}
Under Assumption~\ref{assumption:modes}, there exists $\bar{\rho} > 1$ such that for any threshold $\rho \in (1, \bar{\rho})$, as $N \to \infty$, with probability tending to 1:
\begin{enumerate}[label=(\roman*)]
  \item chains $i, j$ share an edge in $G_\rho$ if and only if their explored regions satisfy $A_i = A_j$;
  \item the connected components of $G_\rho$ partition the chains by their explored mode.
\end{enumerate}
\end{proposition}

\begin{proof}
By (P2)--(P3), 
as $N \to \infty$, 
$\hat{R}_{ij} \xrightarrow{p} 1$ when $\pi_i = \pi_j$, 
and $\hat{R}_{ij} \xrightarrow{p} \lambda_{ij} > 1$ when $\pi_i \neq \pi_j$. 
By Assumption~\ref{assumption:modes}, $\pi_i = \pi_j$ if and only if $A_i = A_j$. 
Let $\bar{\rho} = \min_{A_i \neq A_j} \lambda_{ij}$ be the smallest limiting value over pairs of chains in distinct modes, with $\bar{\rho} = \infty$ if no such pair exists. 
Fix a threshold $\rho \in (1, \bar{\rho})$, 
which lies above the same-mode limit and below every distinct-mode limit.
Since an edge is present if and only if $\hat{R}_{ij} < \rho$, 
the edge $\{i, j\}$ appears if and only if $A_i = A_j$, with probability tending to 1 as $N \to \infty$.
This proves \textit{(i)}.
Part \textit{(ii)} follows from \textit{(i)}:
the connected components partition the chains by their explored region, which by Assumption~\ref{assumption:modes} is a mode.
\end{proof}

\noindent To proceed to mode counting, we introduce the condition that no chain is isolated.

\begin{definition}[Isolation-free graph]
\label{def:isolation-free}
The graph $G_\rho$ is \emph{isolation-free} if $\mathcal{I} = \emptyset$ (equivalently, $I = 0$).
\end{definition}

\begin{corollary}[Mode counting for an isolation-free graph]
\label{cor:mode-count}
Under Assumption~\ref{assumption:modes}, if $G_\rho$ is isolation-free, then $K$ equals the number of distinct modes explored by the chains, each confirmed by at least two chains.
\end{corollary}

\begin{proof}
By Proposition~\ref{prop:asymptotic}, the connected components partition the chains by mode. When the graph is isolation-free there are no isolated chains, so every component lies in $\mathcal{Q}$, and $K = |\mathcal{Q}|$ equals the number of distinct explored modes.
\end{proof}

\subsection{Chain count and mode identification}
\label{sec:chaincountmode}
In practice the per-chain sample size $N$ is finite, and the resulting MCMC convergence graph $G_\rho$ depends  on the number of chains $n$. 
Proposition~\ref{prop:lower-bound} gives a necessary lower bound on $n$.

\begin{proposition}[Lower bound on chain count]
\label{prop:lower-bound}
Under Assumption~\ref{assumption:modes}, 
every mode is the explored region of at least two chains if and only if $K = M$ and $I = 0$, 
which requires the lower bound $n \geq 2M$.
\end{proposition}

\noindent This follows from Proposition~\ref{prop:asymptotic}: the connected components partition the chains by mode, so $K = M$ with $I = 0$ holds exactly when every mode is explored by at least two chains. 
The bound $n \geq 2M$ is then immediate from the pigeonhole principle, and is necessary but not sufficient.
To quantify this, we consider $n$ chains, each independently assigned to one of the $M$ modes with probability $1/M$. A given chain is isolated when all $n - 1$ other chains miss its mode, which occurs with probability $(1 - 1/M)^{n-1}$. 
By linearity of expectation:
\begin{equation}
  \mathbb{E}[I] \;=\; n\left(1 - \tfrac{1}{M}\right)^{n-1}.
  \label{eq:chain_lower_theoretical}
\end{equation}
Reliable mode identification therefore requires $n \gg M$, even in this idealised regime. 
However, in practice each chain carries a non-zero risk $\xi$ of being in disagreement with every other chain, where $\xi$ depends on the posterior geometry and the sampler.
With $n$ independent chains, the probability that at least one chain is isolated is
\begin{equation}
  \mathbb{P}(I \geq 1) \;=\; 1 - (1 - \xi)^n,
  \label{eq:chain_upper_practical}
\end{equation}
which approaches $1$ as $n$ grows. 
Eqs.~\eqref{eq:chain_lower_theoretical} and~\eqref{eq:chain_upper_practical} 
demonstrate a trade-off in choosing $n$ for constructing the MCMC convergence graph.  
Too few chains fail to recover every mode, 
while too many increase the risk of isolated chains due to practical limitations of the sampler: 
\begin{equation}
  \underbrace{\mathbb{E}[I] \xrightarrow{n \to \infty} 0}_{\text{without per-chain risk}},
  \qquad
  \underbrace{\mathbb{P}(I \geq 1) \xrightarrow{n \to \infty} 1}_{\text{with per-chain risk } \xi}.
\end{equation}

\vspace{-.5cm}
\enlargethispage{\baselineskip}
\subsection{Multivariate MCMC convergence graphs}
\label{sec:multidim}
For a multivariate target $\pi$ with dimension $d>1$,  
$d$ pairwise $\hat{R}$ statistics $\hat{R}_{ij,s}$ ($s=1,\ldots,d$) are computed for each chain pair, 
with corresponding per-dimension convergence graphs $G_{\rho,s}$.
One can aggregate these into a single graph by intersection or union of edge sets:
\begin{align}
  G_\cap &= \textstyle\bigcap_{s=1}^d G_{\rho,s}, \tag*{(Intersection)}
    \label{eq:intersection} \\
  G_\cup &= \textstyle\bigcup_{s=1}^d G_{\rho,s}. \tag*{(Union)}
    \label{eq:union}
\end{align}
Intersection requires chain-pairs to agree on every marginal (or dimension),  whereas union requires agreement on at least one.
The set difference $G_\cup \setminus G_\cap$ thus reveals chain pairs that match on some marginals but not others. 
As illustrated in Fig.~\ref{fig:multivariate}, $G_\cap$, $G_\cup$, and their difference $G_\cup \setminus G_\cap$ provide complementary views for diagnosing multimodality. 
The mode identification and chain-count results of Sections~\ref{sec:asymptotic} and~\ref{sec:chaincountmode} 
apply directly to each per-dimension graph $G_{\rho,s}$, 
and extend to the aggregated multivariate setting by replacing $G_\rho$ with $G_\cap$.

\begin{figure}[H]
\centering
\includegraphics[
  width=\textwidth,
]{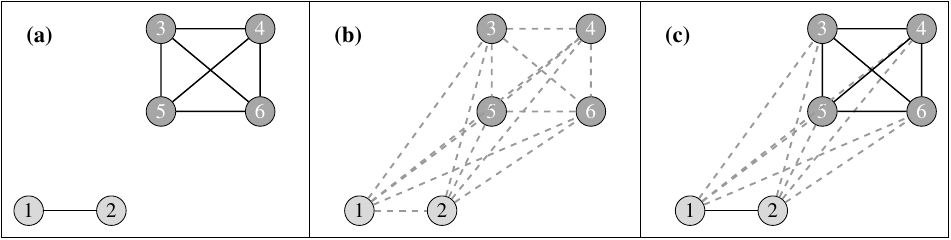}
\caption{
\textbf{Multivariate MCMC convergence graphs constructed from univariate graphs.}
(a) For dimension $s=1$, $G_{\rho,1}$ shows local agreement: 
chains in two connected components ($\{1,2\}$ and $\{3,4,5,6\}$) are in within-component agreement and across-component disagreement.
(b) For dimension $s=2$, $G_{\rho,2}$ shows global agreement: 
all chains form a single connected component.
(c) A multivariate convergence graph is constructed from $G_{\rho,1}$ and $G_{\rho,2}$. 
Solid edges define the intersection $G_\cap$, and dashed edges are the difference $G_\cup \setminus G_\cap$.
}
\label{fig:multivariate}
\end{figure}

\section{Implementation}
\label{sec:implementation}

\subsection{The R package \texttt{mcmcConvergenceGraph}}
\label{sec:workflow}
We implement the MCMC convergence graph diagnostic as an open-source R package~\cite{r2025}, freely available at \url{https://github.com/chain-diagnostics/mcmc-convergence-graph}. 
The package is installed from GitHub via \texttt{devtools::install\_github("chain-diagnostics/mcmc-convergence- graph")}.
It accepts posterior samples from any MCMC sampler and returns both numerical summaries of local chain agreement and a visual representation of the convergence graph. 
Numerical summaries are based on pairwise $\hat{R}$ values calculated as is done in Stan~\cite{vehtari2021}, and the graph is rendered with the \texttt{igraph}~\cite{igraph2026} package. 
The underlying pipeline is outlined in Appendix A.  
\subsection{Using the R package}
\label{sec:using-the-r-package}
The wrapper function \texttt{mcmcgraph()} takes posterior samples as an \texttt{rstan} \texttt{stanfit} object~\cite{rstan2026} or, more generally, as a three-dimensional draws array of shape $N \times n \times d$, where $N$ is the number of post-warmup iterations per chain, $n$ the number of chains, and $d$ the number of monitored parameters. 
Draws are the only required argument. 
Optional arguments include parameters to diagnose, 
the threshold $\rho$ (default 1.05), 
and graph style options. 
A full list of optional input arguments is available on the package's GitHub page. 
The package provides three types of graph outputs: 
per-parameter graphs $G_{\rho, s}$ (one graph per parameter $s$), 
the intersection graph $G_\cap$ (chains agree on every parameter),
and the union graph $G_\cup$ (chains agree on at least one parameter). 
A minimal use case is shown in Code Listing~\ref{lst:mcmcgraph}, 
and package outputs are shown in Section~\ref{sec:results} for the example problems introduced in Section~\ref{sec:experimental_setup}. 
\begin{lstlisting}[style=rlanguage,
caption={Minimal call to run and inspect the MCMC convergence graph diagnostic.},
label=lst:mcmcgraph]
library(mcmcConvergenceGraph)
result <- mcmcgraph(draws)   # stanfit object or a three-dimensional draws array
print(result)
\end{lstlisting}
\vspace{-.5cm}

\subsection{Experimental setup}
\label{sec:experimental_setup}
We apply the MCMC convergence graph diagnostic to three experimental settings of increasing complexity (Fig.~\ref{fig:datagen}). 
In the first setting, chains sample directly from prescribed two-dimensional target distributions with known mode structure. 
In the second, data are generated from two linear regressions with Gaussian noise, and fitted with a single Cauchy regression.
In the third, we fit clinical data to a pharmacokinetic (PK) model. 
All inference is performed in Stan with eight chains, each run for 1000 warmup and 1000 sampling iterations.   Chains are initialised uniformly on the target domain via a custom function to avoid favouring any mode by initialisation. 
We use the threshold $\rho = 1.05$ to construct the graphs. 

\paragraph{(i) Controlled target distributions.}
In a standard Bayesian workflow, the target posterior takes the form
$\pi(\mathbf{x}) \propto p(\mathbf{y} \mid \mathbf{x})\, p(\mathbf{x})$, 
with a likelihood $p(\mathbf{y} \mid \mathbf{x})$ and a prior $p(\mathbf{x})$. 
In this first setting, we bypass this construction and specify the target density directly, providing a controlled test case. 
We consider four target distributions on $[-1, 1]^2$:
\begin{align*}
  \text{Uniform:}              &\quad \pi(\mathbf{x}) \propto 1, \\
  \text{Unimodal Gaussian:}    &\quad \pi = \mathcal{N}(\mathbf{0},\, 0.15\, \mathbf{I}), \\
  \text{Anisotropic Gaussian:} &\quad \pi = \mathcal{N}(\mathbf{0},\, \Sigma),
    \quad \Sigma = \begin{pmatrix} 0.15 & 0.14 \\ 0.14 & 0.15 \end{pmatrix}, \\
  \text{Bimodal Gaussian:}  &\quad \pi = \tfrac{1}{2} \sum_{k=1}^{2} \mathcal{N}(\boldsymbol{\mu}_k,\, 0.008\, \mathbf{I}),
\end{align*}

\noindent with $\boldsymbol{\mu}_1 = (-0.5, 0)$ and $\boldsymbol{\mu}_2 = (0.5, 0)$, as shown in Fig.~\ref{fig:datagen}a. 
Chains are initialised uniformly on $[-1, 1]^2$. 
The diagnostic is applied to each coordinate $x_1, x_2$ separately and combined via the intersection graph $G_\cap$. 

\paragraph{(ii) Cauchy regression on bimodal data.}
We generate $L=300$ observations from two linear regressions with Gaussian noise, 
\begin{equation}
  y_\ell = a_k x_\ell + b_k + \varepsilon_\ell,
  \quad \varepsilon_\ell \sim \mathcal{N}(0, \sigma^2),
\end{equation}
where $k = 1$ for $\ell \leq L/2$ and $k = 2$ otherwise, with $\sigma = 0.05$. 
Data generation is performed in three variations:
\emph{different intercepts} ($a_1 = a_2 = 0$, $b_1 = 2$, $b_2 = -2$),
\emph{different slopes} ($a_1 = 1$, $a_2 = -1$, $b_1 = b_2 = 0$), and
\emph{different intercepts and slopes} ($a_1 = 0.5$, $a_2 = -0.5$, $b_1 = 0.25$, $b_2 = -0.25$),
as shown in Fig.~\ref{fig:datagen}b. 
We fit a single Cauchy regression to each dataset,
\begin{equation}
  y_\ell \sim \mathrm{Cauchy}(a x_\ell + b, \gamma),
\end{equation}
where the scale $\gamma$ is inferred jointly with $(a, b)$. 
This heavy-tailed likelihood allows the bimodal structure of the data to emerge as multimodality in the posterior. 
Chains are initialised uniformly on $[-3, 3]$ for $a$ and $b$, and on $[0, 1]$ for $\gamma$. 
Priors are $a \sim \mathcal{N}(0, 1)$, $b \sim \mathcal{N}(0, 1)$, and $\gamma \sim \mathcal{N}^+(0, 0.2^2)$. 

\begin{figure}[H]
\centering
\includegraphics[width=\textwidth]{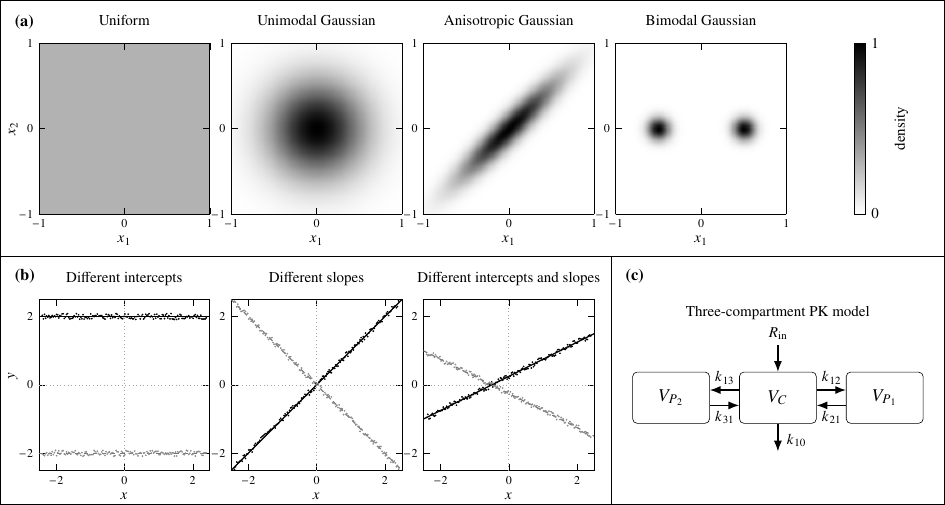}
\caption{
\textbf{Experimental setup of increasing complexity.}
(a) Chains sample directly from one of the four two-dimensional target distributions shown. 
(b) Data drawn from two linear regressions with added Gaussian noise, to be fitted with a single Cauchy regression. 
True lines and $L=300$ illustrative data points per plot are shown. 
(c) Clinical data are fitted to the three-compartment pharmacokinetic (PK) model, shown diagrammatically.
}
\label{fig:datagen}
\end{figure}

\paragraph{(iii) Pharmacokinetic model on clinical data.}
We consider the three-compartment pharmacokinetic (PK) model, commonly used in clinical pharmacology to describe drug distribution and elimination \cite{Minto1997}. 
The model tracks the drug amount in a central compartment $A_C(t)$ and two peripheral compartments $A_1(t), A_2(t)$ over time $t$. 
The dynamics are governed by a system of ordinary differential equations (ODEs),
\begin{equation}
\begin{aligned}
  \frac{dA_C(t)}{dt} &= R_{\text{in}}(t) - (k_{10} + k_{12} + k_{13})\, A_C(t) + k_{21}\, A_1(t) + k_{31}\, A_2(t), \\
  \frac{dA_1(t)}{dt} &= k_{12}\, A_C(t) - k_{21}\, A_1(t), \\
  \frac{dA_2(t)}{dt} &= k_{13}\, A_C(t) - k_{31}\, A_2(t),
\end{aligned}
\label{eq:ode}
\end{equation}

\noindent with initial conditions $A_C(0)=A_1(0)=A_2(0)=0$, assuming no drug is present before dosing begins. 
In Eq.~\eqref{eq:ode}, $R_{\text{in}}(t)$ is the known input rate, $k_{10}$ the elimination rate from the central compartment, and $(k_{12}, k_{21})$, $(k_{13}, k_{31})$ the exchange rates between the central and each peripheral compartment, as shown in Fig.~\ref{fig:datagen}c. 
The observation error follows a combined proportional--additive model with standard-deviation components 
$\sigma_{\text{add}}$ and $\sigma_{\text{prop}}$,
\begin{equation}
  y_{\text{obs}} \sim \mathcal{N}\!\left(y_{\text{pred}},\ (\sigma_{\text{add}} + \sigma_{\text{prop}}\, y_{\text{pred}})^2\right)
  \label{eq:ode_obs}
\end{equation}
where $y_{\text{pred}} = A_C(t) / V_C$ is the drug concentration in the central compartment and $V_C$ is the corresponding volume (Fig.~\ref{fig:datagen}c). 
Analogous relations $A_1(t)/V_{P_1}$ and $A_2(t)/V_{P_2}$ define the (unobserved) peripheral concentrations. 
The ODE model (Eq.~\eqref{eq:ode}) has a known non-identifiability: 
switching $(k_{12}, k_{21}) \leftrightarrow (k_{13}, k_{31})$ leaves the predicted concentrations invariant. 
This symmetry is a form of label switching~\cite{stephens2000}, giving rise to two modes in the posterior.

We fit the full model (Eqs.~\eqref{eq:ode} and~\eqref{eq:ode_obs}) 
to remifentanil concentration data from a previous study of 65 participants~\cite{Minto1997},
with per-participant infusion durations of 4--20 minutes and sampling windows of 44--230 minutes. 
Concentrations are provided in 
$\mathrm{ng}\,\mathrm{mL}^{-1}$, 
rate constants in $\mathrm{min}^{-1}$, 
and volumes in litres (L). 
We implement the ODE model in Stan with complete pooling (following a recent guide to ODE modelling in Stan~\cite{hamis2026ode}), 
inferring a shared parameter set across participants, aiming to demonstrate the diagnostic on a real model with known non-identifiability rather than to draw pharmacological conclusions. 
Priors are based on the study's frequentist results:
$k_{10}, k_{12}, k_{21}, k_{13}, k_{31} \sim \mathrm{LogNormal}(\log(0.3), 1)$
and $V_C \sim \mathrm{LogNormal}(\log(5), 1)$. 
For the errors we use the priors $\sigma_{\text{add}}, \sigma_{\text{prop}} \sim \mathcal{N}^+(0, 1)$.
Chains are initialised uniformly on 
$[0, 2]$ for $k_{10}, k_{12}, k_{21}, k_{13}, k_{31}$, 
on $[0, 10]$ for $V_C$, 
and on $[0, 1]$ for $\sigma_{\text{add}}, \sigma_{\text{prop}}$.

\vspace{1cm}

\section{Results}
\label{sec:results}

\subsection{Controlled target distributions}
\label{sec:results_i}
Numerical summaries from sampling the four controlled target distributions are reported in Table~\ref{tab:res_i}. 
For the uniform and two unimodal Gaussian targets, the Stan-reported $\hat{R}$ and the graph diagnostic agree: 
the former reports $\hat{R} \approx 1$, the latter reports a single connected component ($K = 1$) with no isolated chains ($I = 0$), and every pairwise value stays below the threshold ($\max_Q \hat{R}_{ij} < \rho = 1.05$). 
For the bimodal target, the Stan-reported $\hat{R}$ is large for $x_1$ ($\hat{R} \approx 5.7$), indicating chain disagreement, while $\hat{R} \approx 1$ for $x_2$, the coordinate on which the two modes share the same mean. 
The graph diagnostic resolves this disagreement by characterising connected components: two components of sizes $(5,3)$ for $x_1$, and a single component for $x_2$. Further, the per-component means identify the ground-truth modes.  
Fig.~\ref{fig:res_i} shows the corresponding convergence graphs, which are fully connected for the unimodal targets. 
For the bimodal target, the intersection graph $G_\cap$ partitions into two connected components, matching the two true modes shown in Fig.~\ref{fig:datagen}a.
The union graph $G_\cup$ then draws edges between these components, reflecting that the modes $\boldsymbol{\mu}_1 = (-0.5, 0)$ and $\boldsymbol{\mu}_2 = (0.5, 0)$ coincide on $x_2$ but differ on $x_1$. 

\begin{table}[H]
\centering

\small
\begin{tabular*}{\textwidth}{@{\extracolsep{\fill}} l l c c c l l}
  \toprule
Target & Coordinate & $\hat{R}$ & $K$ & Component & $n_c$ & $\max_Q \hat{R}_{ij}$ \quad Mean [95\% CI] \\
\midrule 

\multirow{2}{*}{Uniform}
  & $x_1$ & 1.001 & 1 & 1 & 8 & 1.005 \quad $-$0.003 [$-$0.947, $\phantom{-}$0.948] \\
\arrayrulecolor{gray!40}
\cmidrule{2-7}
\arrayrulecolor{black}
  & $x_2$ & 1.000 & 1 & 1 & 8 & 1.002 \quad $\phantom{-}$0.001 [$-$0.956, $\phantom{-}$0.952] \\
\midrule

\multirow{2}{*}{\shortstack[l]{Unimodal\\ Gaussian}}
  & $x_1$ & 1.000 & 1 & 1 & 8 & 1.006 \quad $\phantom{-}$0.004 [$-$0.727, $\phantom{-}$0.741] \\
\arrayrulecolor{gray!40}
\cmidrule{2-7}
\arrayrulecolor{black}
  & $x_2$ & 1.000 & 1 & 1 & 8 & 1.004 \quad $-$0.001 [$-$0.750, $\phantom{-}$0.732] \\
\midrule

\multirow{2}{*}{\shortstack[l]{Anisotropic\\ Gaussian}}
  & $x_1$ & 1.005 & 1 & 1 & 8 & 1.016 \quad $\phantom{-}$0.005 [$-$0.704, $\phantom{-}$0.703] \\
\arrayrulecolor{gray!40}
\cmidrule{2-7}
\arrayrulecolor{black}
  & $x_2$ & 1.005 & 1 & 1 & 8 & 1.017 \quad $\phantom{-}$0.004 [$-$0.695, $\phantom{-}$0.702] \\
\midrule

\multirow{3}{*}{\shortstack[l]{Bimodal\\ Gaussian}}
  & \multirow{2}{*}{$x_1$} & \multirow{2}{*}{5.692} & \multirow{2}{*}{2}
    & 1 & 5 & 1.004 \quad $\phantom{-}$0.499 [$\phantom{-}$0.326, $\phantom{-}$0.673] \\
  & & & & 2 & 3 & 1.002 \quad $-$0.501 [$-$0.675, $-$0.329] \\
\arrayrulecolor{gray!40}
\cmidrule{2-7}
\arrayrulecolor{black}
  & $x_2$ & 1.000 & 1 & 1 & 8 & 1.002 \quad $-$0.001 [$-$0.175, $\phantom{-}$0.175] \\
\bottomrule

\end{tabular*}

\vspace{.1cm}
\caption{
{\bf Diagnostic outputs for the four controlled target distributions.}
For each target and coordinate ($x_1, x_2$) we show the Stan-reported $\hat{R}$, the number of connected components $K$ in the MCMC convergence graph (with threshold $\rho = 1.05$), and per-component summaries: chain count $n_c$, maximum within-component pairwise value $\max_Q \hat{R}_{ij}$, and the target mean  with $95\%$ credible interval.
}
\label{tab:res_i}
\end{table}

\begin{figure}[H]
\centering
\includegraphics[width=\textwidth]{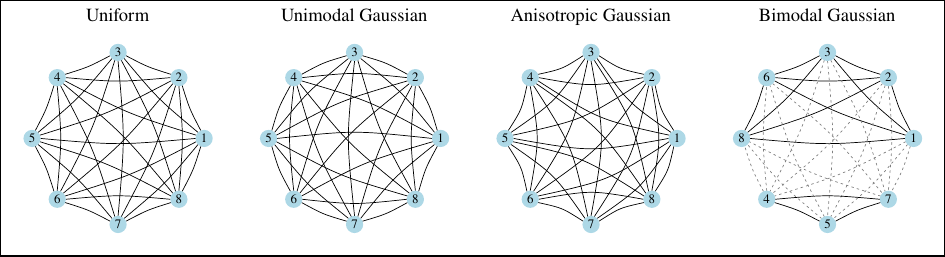}
\caption{
{\bf MCMC convergence graphs for the four controlled target distributions.}
Solid edges connect chain pairs that agree on both coordinates and thus form the intersection graph $G_\cap$. 
Dashed edges connect pairs that agree on one but not both coordinates. 
Together, solid and dashed edges form the union graph $G_\cup$.
The uniform and two unimodal targets yield graphs with a single connected component, indicating global chain agreement.
The bimodal target's intersection graph has two connected components 
$(\{1,2,3,6,8\}, \{4,5,7\})$, indicating local chain agreement. 
The union graph further shows that these components agree on one coordinate ($x_2$).
}
\label{fig:res_i}
\end{figure}

\subsection{Cauchy regression on bimodal data}
\label{sec:results_ii}
Numerical summaries for the three Cauchy regressions on the bimodal data are reported in Table~\ref{tab:res_ii}.
The Stan-reported $\hat{R}$ is large ($\gg 1$) precisely for the parameters whose data were generated from two distinct ground-truth values: 
$b$ under different intercepts but the same slope ($\hat{R} \approx 41$), $a$ under different slopes but the same intercept ($\hat{R} \approx 45$), and both $a$ and $b$ under different slopes and intercepts ($\hat{R} \approx 26$). 
This data-driven bimodality is captured by the graph diagnostic, which partitions the chains into two connected components ($K = 2$) in each of these cases.
Moreover, the per-component means recover the underlying values. 
For the parameters with data generated from a single value, both the Stan-reported $\hat{R}$ and the maximum pairwise $\hat{R}_{ij}$ within the fully connected single component ($K = 1$) are approximately~1. 
The same holds for the scale parameter $\gamma$. 
The convergence graphs for the three Cauchy regressions are shown in Fig.~\ref{fig:res_ii}a, and the posterior for the variant with different slopes and intercepts is shown in Fig.~\ref{fig:res_ii}b, illustrating identification of the two modes.
\begin{table}[H]
\centering
\small
\begin{tabular*}{\textwidth}{@{\extracolsep{\fill}} l l c c c c l}
  \toprule
Model & Parameter & $\hat{R}$ & $K$ & Component & $n_c$ & $\max_Q \hat{R}_{ij}$ \quad Mean [95\% CI] \\
\midrule
\multirow{4}{*}{\shortstack[l]{Different intercepts}}
  & $a$ & 1.004 & 1 & 1 & 8 & 1.011 \quad $-$0.005 [$-$0.112, $\phantom{-}$0.104] \\
\arrayrulecolor{gray!40}
\cmidrule{2-7}
\arrayrulecolor{black}
  & \multirow{2}{*}{$b$} & \multirow{2}{*}{41.448} & \multirow{2}{*}{2}
    & 1 & 4 & 1.016 \quad $\phantom{-}$1.930 [$\phantom{-}$1.810, $\phantom{-}$2.004] \\
  & & & & 2 & 4 & 1.006 \quad $-$1.926 [$-$1.997, $-$1.818] \\
\arrayrulecolor{gray!40}
\cmidrule{2-7}
\arrayrulecolor{black}
  & $\gamma$ & 1.008 & 1 & 1 & 8 & 1.026 \quad $\phantom{-}$0.501 [$\phantom{-}$0.303, $\phantom{-}$0.765] \\
\midrule
\multirow{4}{*}{\shortstack[l]{Different slopes}}
  & \multirow{2}{*}{$a$} & \multirow{2}{*}{45.039} & \multirow{2}{*}{2}
    & 1 & 5 & 1.001 \quad $-$0.951 [$-$0.989, $-$0.904] \\
  & & & & 2 & 3 & 1.003 \quad $\phantom{-}$0.958 [$\phantom{-}$0.911, $\phantom{-}$0.995] \\
\arrayrulecolor{gray!40}
\cmidrule{2-7}
\arrayrulecolor{black}
  & $b$ & 1.009 & 1 & 1 & 8 & 1.017 \quad $\phantom{-}$0.002 [$-$0.020, $\phantom{-}$0.023] \\
\arrayrulecolor{gray!40}
\cmidrule{2-7}
\arrayrulecolor{black}
  & $\gamma$ & 1.004 & 1 & 1 & 8 & 1.010 \quad $\phantom{-}$0.164 [$\phantom{-}$0.126, $\phantom{-}$0.210] \\
\midrule
\multirow{5}{*}{\shortstack[l]{Different slopes\\ and intercepts}}
  & \multirow{2}{*}{$a$} & \multirow{2}{*}{26.044} & \multirow{2}{*}{2}
    & 1 & 6 & 1.004 \quad $-$0.455 [$-$0.484, $-$0.423] \\
  & & & & 2 & 2 & 1.001 \quad $\phantom{-}$0.456 [$\phantom{-}$0.422, $\phantom{-}$0.485] \\
\arrayrulecolor{gray!40}
\cmidrule{2-7}
\arrayrulecolor{black}
  & \multirow{2}{*}{$b$} & \multirow{2}{*}{25.877} & \multirow{2}{*}{2}
    & 1 & 6 & 1.004 \quad $-$0.241 [$-$0.257, $-$0.224] \\
  & & & & 2 & 2 & 1.000 \quad $\phantom{-}$0.246 [$\phantom{-}$0.228, $\phantom{-}$0.262] \\
\arrayrulecolor{gray!40}
\cmidrule{2-7}
\arrayrulecolor{black}
  & $\gamma$ & 1.002 & 1 & 1 & 8 & 1.015 \quad $\phantom{-}$0.111 [$\phantom{-}$0.091, $\phantom{-}$0.135] \\
\bottomrule
\end{tabular*}

\caption{
{\bf Diagnostic outputs for the three Cauchy regressions on bimodal data.}
For each model and parameter we show the Stan-reported $\hat{R}$, the number of connected components $K$ in the MCMC convergence graph (with threshold $\rho = 1.05$), and per-component summaries: chain count $n_c$, maximum within-component pairwise value $\max_Q \hat{R}_{ij}$, and the posterior mean with $95\%$ credible interval.
}
\label{tab:res_ii}

\end{table}

\begin{figure}[H]
\centering
\includegraphics[
  width=\textwidth
]{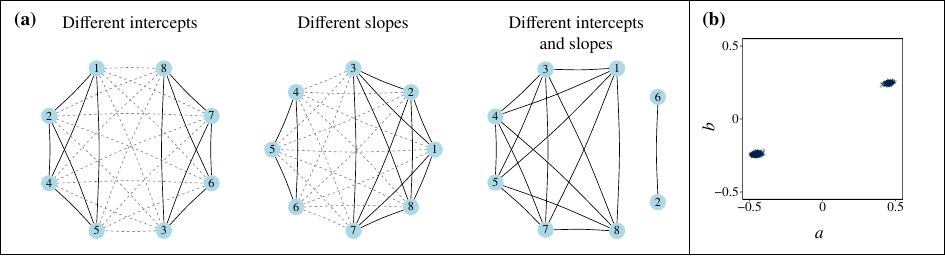}
\caption{
{\bf MCMC convergence graphs for Cauchy regression on bimodal data.}
(a) Solid edges connect chain pairs that agree on both monitored parameters and thus form the intersection graph $G_\cap$. Dashed edges connect pairs that agree on only one parameter (parameter $a$ when the data is generated from linear regression with different intercepts but the same slope; $b$ with different slopes but the same intercept). Together, solid and dashed edges form the union graph $G_\cup$.
(b) The posterior resulting from linear regression-generated data with two different intercepts and slopes.
}
\label{fig:res_ii}
\end{figure}

\enlargethispage{\baselineskip}
\subsection{Pharmacokinetic model on clinical data}
\label{sec:results_iii}
Numerical summaries for the three-compartment PK model are reported in Table~\ref{tab:res_iii}. 
The table shows that the elimination rate $k_{10}$ and central-compartment volume $V_C$, 
as well as the error parameters $\sigma_{\text{add}}$ and $\sigma_{\text{prop}}$, 
have Stan-reported $\hat{R}$ values close to 1.
In line with this, a single connected component ($K = 1$) is identified on these parameters. 
The four exchange rates $k_{12}, k_{21}, k_{13}, k_{31}$, by contrast, are flagged by large Stan-reported values ($\hat{R} \gg 1$). 
Our graph diagnostic identifies two connected components on these parameters ($K = 2$), 
and inspection of the inferred parameter values demonstrates a label-switching symmetry: the values taken by $(k_{12}, k_{21})$ in one component are taken by $(k_{13}, k_{31})$ in the other, and vice versa. 
This is illustrated in Fig.~\ref{fig:res_iii} and makes these parameters non-identifiable, 
as the switch leaves the predicted concentrations unchanged.


\begin{table}[H]
\centering
\small
\begin{tabular*}{\textwidth}{@{\extracolsep{\fill}} l l c c c c l}
  \toprule
Model & Parameter & $\hat{R}$ & $K$ & Component & $n_c$ &  $\max_Q \hat{R}_{ij}$  \quad Mean [95\% CI] \\
\midrule
\multirow{12}{*}{\shortstack[l]{Three-compartment\\ PK model}}
  & $k_{10}$ & 1.002 & 1 & 1 & 8 & 1.006 \quad $\phantom{-}$0.385 [$\phantom{-}$0.363, $\phantom{-}$0.409] \\
  \arrayrulecolor{gray!40}
  \cmidrule{2-7}
  \arrayrulecolor{black}
  & \multirow{2}{*}{$k_{12}$} & \multirow{2}{*}{7.991} & \multirow{2}{*}{2}
    & 1 & 3 & 1.001 \quad $\phantom{-}$0.014 [$\phantom{-}$0.012, $\phantom{-}$0.017] \\
  & & & & 2 & 5 & 1.007 \quad $\phantom{-}$0.214 [$\phantom{-}$0.185, $\phantom{-}$0.247] \\
  \arrayrulecolor{gray!40}
  \cmidrule{2-7}
  \arrayrulecolor{black}
  & \multirow{2}{*}{$k_{21}$} & \multirow{2}{*}{14.459} & \multirow{2}{*}{2}
    & 1 & 3 & 1.001 \quad $\phantom{-}$0.016 [$\phantom{-}$0.012, $\phantom{-}$0.020] \\
  & & & & 2 & 5 & 1.006 \quad $\phantom{-}$0.167 [$\phantom{-}$0.155, $\phantom{-}$0.180] \\
  \arrayrulecolor{gray!40}
  \cmidrule{2-7}
  \arrayrulecolor{black}
  & \multirow{2}{*}{$k_{13}$} & \multirow{2}{*}{10.366} & \multirow{2}{*}{2}
    & 1 & 3 & 1.001 \quad $\phantom{-}$0.213 [$\phantom{-}$0.185, $\phantom{-}$0.245] \\
  & & & & 2 & 5 & 1.003 \quad $\phantom{-}$0.015 [$\phantom{-}$0.012, $\phantom{-}$0.017] \\
  \arrayrulecolor{gray!40}
  \cmidrule{2-7}
  \arrayrulecolor{black}
  & \multirow{2}{*}{$k_{31}$} & \multirow{2}{*}{17.688} & \multirow{2}{*}{2}
    & 1 & 3 & 1.004 \quad $\phantom{-}$0.167 [$\phantom{-}$0.155, $\phantom{-}$0.179] \\
  & & & & 2 & 5 & 1.003 \quad $\phantom{-}$0.016 [$\phantom{-}$0.012, $\phantom{-}$0.020] \\
  \arrayrulecolor{gray!40}
  \cmidrule{2-7}
  \arrayrulecolor{black}
  & $V_C$ & 1.002 & 1 & 1 & 8 & 1.005 \quad $\phantom{-}$6.665 [$\phantom{-}$6.241, $\phantom{-}$7.070] \\
  \arrayrulecolor{gray!40}
  \cmidrule{2-7}
  \arrayrulecolor{black}
  & $\sigma_{\mathrm{add}}$ & 1.002 & 1 & 1 & 8 & 1.006 \quad $\phantom{-}$0.049 [$\phantom{-}$0.034, $\phantom{-}$0.065] \\
  \arrayrulecolor{gray!40}
  \cmidrule{2-7}
  \arrayrulecolor{black}
  & $\sigma_{\mathrm{prop}}$ & 1.001 & 1 & 1 & 8 & 1.006 \quad $\phantom{-}$0.299 [$\phantom{-}$0.287, $\phantom{-}$0.311] \\
\bottomrule
\end{tabular*}

\vspace{.1cm}
\caption{
{\bf Diagnostic outputs for the three-compartment pharmacokinetic model.}
For each parameter we show the Stan-reported $\hat{R}$, the number of connected components $K$ in the MCMC convergence graph (with threshold $\rho = 1.05$), and per-component summaries: chain count $n_c$, maximum within-component pairwise value $\max_Q \hat{R}_{ij}$, and the posterior mean with $95\%$ credible interval.
Rate constants are in $\mathrm{min}^{-1}$, 
$V_C$ in L, $\sigma_{\text{add}}$ in $\mathrm{ng}\,\mathrm{mL}^{-1}$, and $\sigma_{\text{prop}}$ is dimensionless. 
}
\label{tab:res_iii}
\end{table}
\begin{figure}[H]
\centering
\includegraphics[
  width=\textwidth
]{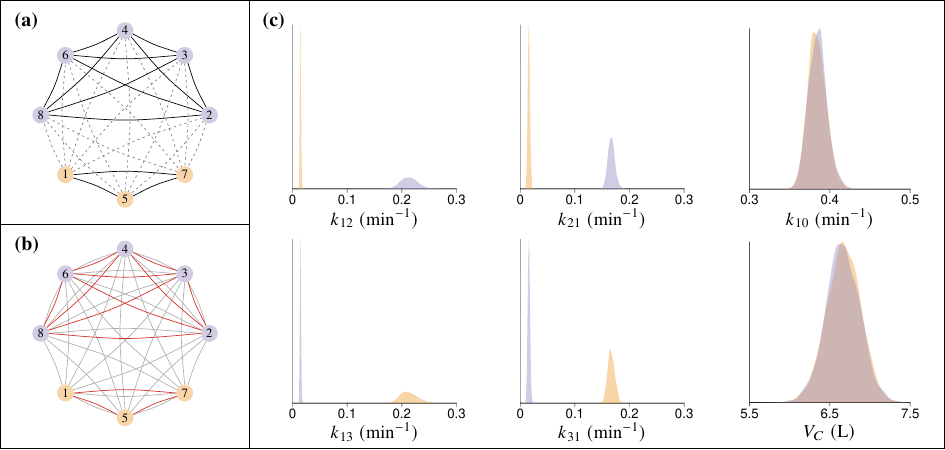}
\caption{
{\bf MCMC convergence graphs and posteriors for the three-compartment pharmacokinetic model.}
(a) The intersection graph $G_\cap$ (solid lines) has two connected components (\{2,3,4,6,8\} and \{1,5,7\}) whose chains agree on all parameters.  
Nodes are coloured by connected component. 
The union graph $G_\cup$ (solid and dashed lines together) is fully connected, as every pair of chains agrees on at least one parameter.  
(b) Per-parameter graphs for $V_C$ (grey lines) and $k_{31}$ (red lines) with overlaid edges, demonstrating global agreement in $V_C$ and local agreement in $k_{31}$. 
(c) Marginal posterior densities, coloured by connected components in $G_\cap$.
}
\label{fig:res_iii}
\end{figure}

\section{Discussion}
\label{sec:discussion}
Having established the theoretical basis of the MCMC convergence graph (Section~\ref{sec:theory}) and empirically demonstrated its use on a selection of worked examples (Section~\ref{sec:results}), we now discuss practical considerations of the diagnostic.

\subsection{Chain initialisation and count selection}
The MCMC convergence graph can only reveal modes that at least one chain visits. To recover every mode, therefore, chains must be initialised broadly across the parameter space. Random initialisation from a broad prior, or optimisation-based initialisation such as Pathfinder~\cite{zhang2022pf}, reduces the risk that a mode is left unvisited.
Because the chain-count choice trades off mode coverage against spurious isolation (Section~\ref{sec:chaincountmode}), we recommend running the diagnostic across a sweep of chain counts $n$ (see Appendix B.1 for an example). For each monitored parameter, plotting $K$ and $I/n$ against $n$ reveals the diagnostic's stability: a plateau in $K$ indicates reliable mode identification, whereas persistent changes signal incomplete mode coverage or unstable pairwise $\hat{R}_{ij}$ estimates. 
When the chain count substantially exceeds the identified component count ($n \gg K$), an isolated chain is more likely to reflect finite-sample noise than a mode explored by a single chain.
This effect grows with the per-chain risk $\xi$ (Section~\ref{sec:chaincountmode}) and is therefore more pronounced for difficult posterior geometries. 
Focused inspection of the outlier chain (its trace, split-$\hat{R}$, and per-parameter pairwise $\hat{R}_{ij}$ values) can further distinguish spurious isolation from a chain that has genuinely failed to mix.

\subsection{Selecting the threshold $\rho$}
\label{sec:disc_rho}
The threshold $\rho$ controls the granularity of the resulting graph.
In the theoretical limit ($N \to \infty$, Assumption~\ref{assumption:modes}), the pairwise statistic converges to 1 for chains exploring the same mode and to a pair-specific limit $\lambda_{ij} > 1$ otherwise, so any threshold $\rho \in (1, \bar{\rho})$ recovers the correct partition (Proposition~\ref{prop:asymptotic}).
In this idealised regime chain agreement is transitive: 
if chains $i$ and $j$ share an edge and chains $j$ and $k$ share an edge, 
then $i$ and $k$ share an edge. 
For finite $N$, however, $\hat{R}_{ij}$ is a noisy estimate, and computed values near $\rho$ may fall on either side of it. Transitivity can then fail, but the identified connected-component structure remains informative, since chains exploring the same mode are typically linked through a path of intermediate agreements. 

Because each pairwise $\hat{R}_{ij}$ is computed from only two chains, individual pairwise values can be larger than a standard $\hat{R}$ computed from more chains, even for well-mixed unimodal posteriors (see, e.g., the unimodal cases in Table~\ref{tab:res_i}). This motivates using a slightly relaxed convergence threshold. 
The choice of $\rho$ trades resolution against robustness: a small threshold (e.g., $\rho = 1.01$) resolves finer structure but is more prone to non-transitivity and spurious isolated chains, whereas a large threshold (e.g., $\rho = 1.1$) may merge distinct modes. 
When robustness is in doubt, we recommend repeating the analysis across a range of thresholds and confirming that $K$ is stable (see Appendix B.2 for an example). 
As this is done entirely post-sampling, it is inexpensive.

\subsection{Acting on the graph output}
\label{sec:sub_act_on_graph}
When the classical (e.g., Stan-reported) $\hat{R}$ and the graph-based diagnostic agree, interpretation is straightforward: an $\hat{R}$ close to 1 with a single component ($K = 1$) supports reliable inference, whereas an elevated $\hat{R}$ with widespread isolation ($I/n$ large) calls for standard troubleshooting. The graph adds diagnostic value when the two disagree. 
When a classical $\hat{R}$ value is elevated but the corresponding pairwise $\hat{R}_{ij}$ values within each connected component remain low, chains have partitioned into groups that each mix well within their posterior region but disagree across regions. 
The appropriate follow-up depends on the origin of the local chain agreement, specifically whether it arises from features of the data or of the model. 
In the former case, identifying the full posterior is often desirable, and the graph helps us characterise it without discarding the inference as a sampling failure (see, e.g., Fig.~\ref{fig:res_ii}). 
In the latter case, the graph helps us detect non-identifiability in the model structure (see, e.g., Fig.~\ref{fig:res_iii}), which gives insight into model behaviour and can, if desired, be addressed through model reformulation, reparameterisation, or stricter priors.

\newpage

\subsection{Mode partitioning versus mode weighting}
The proportion of chains in each connected component depends on chain initialisation, the sampler used, and posterior geometry. It should therefore not be interpreted as the posterior weight of the corresponding region. This is an important limitation of the MCMC convergence graph.

Estimating posterior weights from non-mixing MCMC is a recognised difficulty. A pragmatic alternative is Bayesian stacking~\cite{yao2022}, which combines the inferences from groups of chains using weights chosen to optimise predictive accuracy rather than to estimate posterior mass. The convergence graph and stacking therefore play complementary roles: the former partitions chains by their explored posterior region, whereas the latter aggregates their inferences for prediction.

\subsection{Scaling to high-dimensional posteriors}
In theory, the MCMC convergence graph extends to arbitrary dimension. 
Because the graph is computed post-sampling, and sampling dominates the overall computational cost, the main practical concern is the sampling budget rather than the graph computation. 
However, as the dimension $d$ grows, so too does the risk that global agreement fails in at least one dimension, and thus we can expect $K_{\max} = \max_s K_s$ to grow with $d$. 
Since the multivariate intersection graph $G_\cap$ is at least as fragmented as the most partitioned dimension, reliable identification requires $n \geq 2 K_{\max}$ chains by Proposition~\ref{prop:lower-bound}. 
In other words, larger dimension may require more chains, and thus more computation. 

Visual diagnostics also become inherently less informative at large dimension. 
In practice one can isolate marginals for closer inspection, for example by focusing on the parameters of particular interest or those where global chain agreement fails. 
One can also contract each connected component of $G_\cap$ to a single node, joined to another wherever the two agree on some dimension in $G_\cup$. 
This gives a coarser graph that summarises how the modes relate and remains legible as $d$ grows. 
Even so, it is the visual assessment that does not scale, not the graph structure itself: the same components, $K$, and $I$ are equally available in tabular form, which is often the clearer view in high dimension. 

\subsection{Extending the pairwise statistic to $k$-wise statistics}
The MCMC convergence graph presented in this work builds on pairwise $\hat{R}_{ij}$ values between two chains. 
Since standard convergence diagnostics typically recommend at least four chains~\cite{stan2026}, extending to $k$-wise statistics is a natural generalisation, with properties (P1)--(P3) from Section~\ref{sec:pairwise-rhat} carrying over directly to $k$-wise analogues (P1$'$)--(P3$'$):

\begin{enumerate}[label=(P\arabic*$'$),leftmargin=*,labelindent=\parindent,itemsep=2pt]
\item \emph{Symmetry.} $\hat{R}_{i_1,\ldots,i_k}$ is invariant under permutations of the indices.
\item \emph{Agreement.} $\hat{R}_{i_1,\ldots,i_k} \xrightarrow{p} 1$ as $N \to \infty$ if $\pi_{i_1} = \cdots = \pi_{i_k}$.
\item \emph{Disagreement.} $\hat{R}_{i_1,\ldots,i_k} \xrightarrow{p} \lambda_{i_1,\ldots,i_k} > 1$ as $N \to \infty$ otherwise.
\end{enumerate}

\noindent We adopt pairwise ($k=2$) as the primitive throughout the paper, chosen for three reasons.
First, chain agreement is transitive in the limit $N \to \infty$ (Section \ref{sec:disc_rho}). In this limit, pairwise checks are therefore sufficient to recover the partition into connected components, and higher-order checks do not refine it. 
Second, the pairwise structure induces an ordinary graph rather than a hypergraph, which enables direct visualisation and standard connected-component analysis. 
Third, although the two-chain $\hat{R}_{ij}$ is more variable than a standard multi-chain $\hat{R}$, this variability is offset by the aggregate structure of the graph, since each chain participates in $n-1$ pairwise comparisons.
A systematic study of $k$-wise variants and hypergraph diagnostics is left to future work.

\newpage

\section{Conclusion}
\label{sec:conclusion}
We introduced the MCMC convergence graph, a diagnostic that turns pairwise $\hat{R}_{ij}$ values into an undirected graph whose connected components reveal which chains explore the same posterior mode, with theoretical guarantees when the modes are disjoint (Assumption~\ref{assumption:modes}).
Where a classical $\hat{R}$ returns a single verdict of non-convergence, the graph characterises the structure behind it, separating chains that have partitioned into distinct modes from a chain that has failed to mix. 
This distinction is actionable.
If the local agreement comes from the data, the modes may correspond to distinct regimes in the underlying system, be it biological, economic, or social, and are worth understanding in their own right. 
On the other hand, if it comes from the model structure, it reflects a non-identifiability that can be addressed through reparameterisation or stricter priors.
Such a non-identifiability is also informative, giving insight into the model's structure.

Uncovering multimodality by pairwise chain analysis, the MCMC convergence graph complements classical convergence tools and is freely available in the R package \texttt{mcmcConvergenceGraph}, which takes MCMC draws from any sampler and returns both numerical and graphical summaries.
It allows us to diagnose multimodality rather than dismiss it as a sampling failure.

\section*{Funding}
This work was partially supported by the Wallenberg AI, Autonomous Systems and
Software Program (WASP) funded by the Knut and Alice Wallenberg Foundation, which supported all authors. 
EL was also supported by grants from Chalmers Area of Advance Health Engineering, Gender Initiative for Excellence, Ragnar Söderbergs Foundation, and the Swedish Research Council (2024-04145).
SH was also supported by the Swedish Research Council (2024-05621), the Wenner-Gren Foundations (WGF2022-0044), and the Kjell och M{\"a}rta Beijer Foundation. 

\section*{Code and data availability}
All code and data is freely available on the project's GitHub page: 
\url{https://github.com/chain-diagnostics/mcmc-convergence-graph}. 
The results in this paper use version~1.0.0; 
the latest release is available at the same address.

\section*{Author contributions}
\textbf{CCG}: Software, Investigation, Writing -- original draft, review \& editing. 
\textbf{EL}: Supervision, Funding acquisition, Writing -- review \& editing. 
\textbf{SH}: Conceptualisation, Methodology, Supervision, Funding acquisition, Writing -- original draft, review \& editing. 

\section*{AI use statement}
The authors used AI to edit prose and figures. All conceptual, methodological, theoretical, and empirical contributions are the authors' own.

\printbibliography
\newpage

\appendix
\setcounter{figure}{0}
\setcounter{page}{1}
\renewcommand{\thefigure}{A\arabic{figure}}

\section*{Appendix A: The \texttt{mcmcConvergenceGraph} pipeline}
\label{app:package}
The package runs the diagnostic in four stages, shown in Fig.~\ref{fig:workflow}. 
The individual functions at each stage are documented on the project's GitHub page.
The wrapper \texttt{mcmcgraph()} runs all four stages in a single call (Listing~\ref{lst:mcmcgraph}).
The package builds on base R~\cite{r2025}, and graphs are rendered with \texttt{igraph}~\cite{igraph2026}. 

\begin{figure}[H]
\centering
\includegraphics[width=\textwidth]{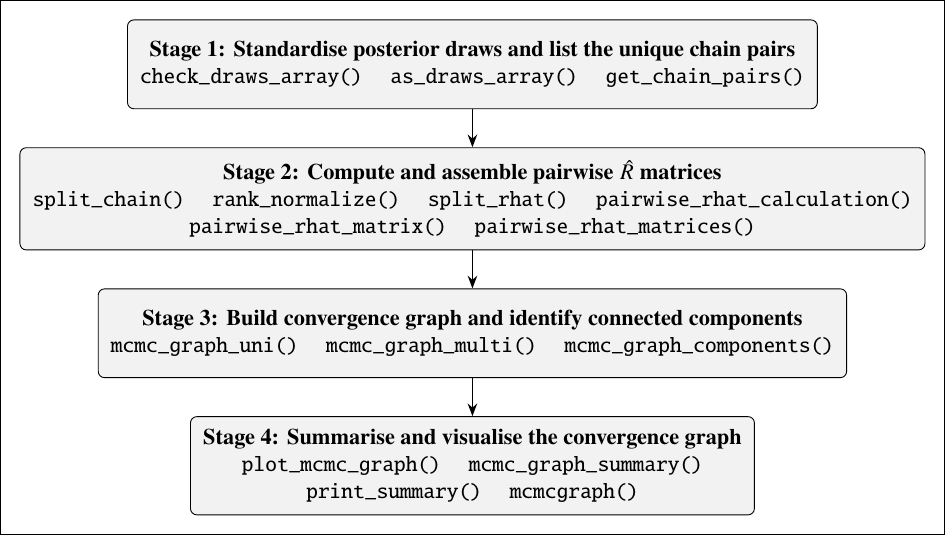}
\caption{
\textbf{The four-stage \texttt{mcmcConvergenceGraph} workflow.}
Stage~1 organises posterior draws into a standardised array and lists the unique chain pairs.
Stage~2 computes the pairwise $\hat{R}_{ij,s}$ diagnostic for every chain pair $\{i,j\}$ and parameter $s$, and assembles these values into $d$ symmetric $n \times n$ matrices, where $d$ is the number of monitored parameters and $n$ the number of chains.
Stage~3 applies the threshold $\rho$ to build the per-parameter graphs $G_{\rho,s}$ and their combinations ($G_\cap$, $G_\cup$), then identifies the connected components.
Stage~4 returns numerical summaries and a diagnostic plot to the user, on screen or as files.
The function names and flow correspond to package version~1.0.0; 
see the GitHub page for the latest version.
}

\label{fig:workflow}
\end{figure}

\section*{Appendix B: Sensitivity analyses}
\label{app:sensitivity}
We assess sensitivity to the chain count $n$ and threshold $\rho$ on a trimodal Gaussian target:
\begin{align*}
  \text{Trimodal Gaussian:} &\quad \pi = \tfrac{1}{3} \sum_{k=1}^{3} \mathcal{N}(\boldsymbol{\mu}_k,\, \sigma^2\, \mathbf{I}),
\end{align*}
with $\boldsymbol{\mu}_1 = (0.5, 0.5)$, $\boldsymbol{\mu}_2 = (-0.5, 0.5)$, and
$\boldsymbol{\mu}_3 = (0, -0.5)$ (Fig.~\ref{fig:appendix_sensitivity}a).
Chains are run in Stan (1000 warmup and 1000 sampling iterations) from uniform initialisation on $[-1, 1]^2$.

\paragraph{Appendix B1: Sensitivity to the chain count $n$}
We perform independent Stan fits at $\sigma^2 = 0.008$ and $\rho = 1.05$ for chain counts
$n = 2, \dots, 100$, repeating each with 10 random seeds.
Results are shown in Fig.~\ref{fig:appendix_sensitivity}b.
Plotting $K$ against $n$ shows that mode identification stabilises once enough chains are used:
$K = 3$ for $x_1$ from $n = 23$ and $K = 2$ for $x_2$ from $n = 7$, across all seeds, and is
unstable below these counts.
Plotting $I/n$ against $n$ shows that isolation is rare beyond small $n$: $I/n = 0$ for $x_2$
from $n = 7$, and near 0 for $x_1$, where an isolated chain still appears spuriously at larger $n$.
Fig.~\ref{fig:appendix_sensitivity}c shows the convergence graph at $n = 24$ for one seed.

\newpage

\paragraph{Appendix B2: Sensitivity to the threshold $\rho$}
We fix the chain count at $n = 24$ and perform one Stan fit at each of two within-mode variances, 
$0.008$ and $0.0114$ (Fig.~\ref{fig:appendix_sensitivity}d).
In both cases the correct mode count is recovered for all $\rho \geq 1.01$ ($K = 3$ for $x_1$ and $K = 2$ for $x_2$), so setting $\rho = 1.05$ yields robust mode identification.
The correct $K$ can appear already at $\rho = 1$ because, just like a classical $\hat{R}$, a pairwise $\hat{R}_{ij}$ can fall below 1 in finite samples.
The isolation plots show that $I/n = 0$ from $\rho \geq 1.01$ for the lower-variance target,
whereas the higher-variance target requires $\rho > 1.11$ to reach $I/n = 0$ on both coordinates. 
For the latter target, $\rho = 1.05$ is therefore not sufficient for an isolation-free graph. 
Instead, a higher threshold is required, illustrating a resolution--robustness trade-off.

\begin{figure}[H]
\centering
\includegraphics[
  width=\textwidth,
]{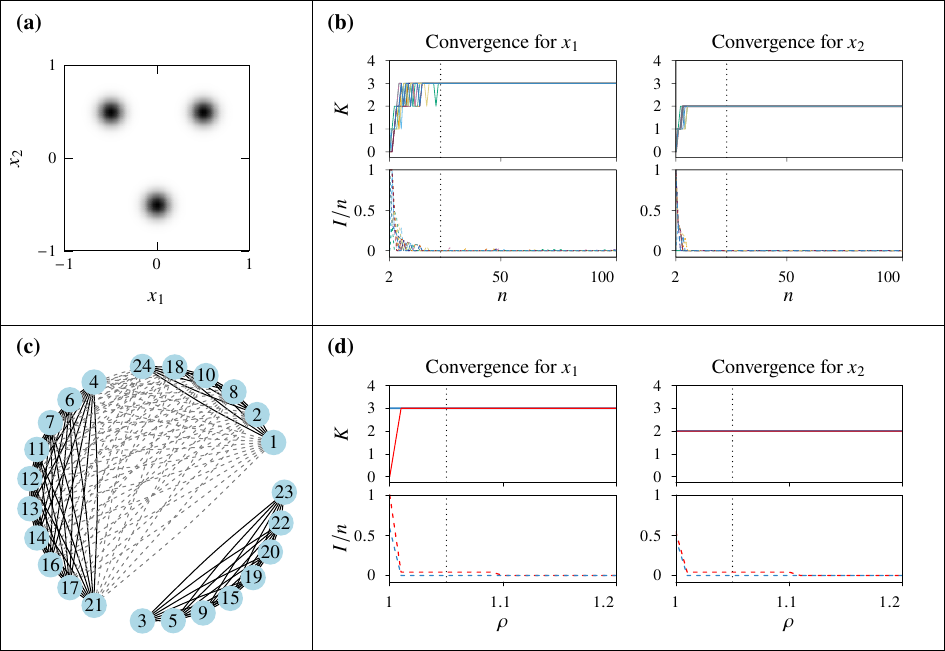}
\caption{
{\bf MCMC convergence graph sensitivity to chain count $n$ and threshold $\rho$.}
(a) Trimodal Gaussian target distribution with within-mode variance $\sigma^2 = 0.008$.
(b) $K$ (top) and $I/n$ (bottom) versus chain count $n$ at $\rho = 1.05$.
Each coloured line shows the result for one random seed, and the dotted vertical line marks $n = 24$. 
(c) Convergence graph from a single run at $n = 24$ ($\sigma^2 = 0.008$, $\rho = 1.05$). 
Solid edges connect chain pairs that agree on both coordinates (the intersection graph $G_\cap$), dashed edges those that agree only on $x_2$, and together they form the union graph $G_\cup$.
(d) $K$ (top) and $I/n$ (bottom) versus threshold $\rho$ at $n = 24$. 
Results for $\sigma^2 = 0.008$ (blue) and $\sigma^2 = 0.0114$ (red) are shown, with the dotted vertical line marking $\rho = 1.05$.
}

\label{fig:appendix_sensitivity}
\end{figure}

\end{document}